\documentclass[11pt]{llncs}
\usepackage{microtype}
\usepackage{makeidx}
\usepackage{graphicx}
\usepackage{xcolor}
\usepackage{ifthen}
\usepackage{microtype}
\usepackage{boxedminipage}
\usepackage{amsmath}
\usepackage{amsfonts}
\usepackage{amssymb,mathrsfs}
\usepackage{authblk}
\usepackage{breakcites}
\usepackage{chemscheme}
\usepackage{hyperref}
\usepackage{framed}
\usepackage{float}
\usepackage{breqn}
\usepackage{array, makecell}

\newcommand{\lv}[1]{}

\newcommand{\veps}{\varepsilon}
\newcommand{\Cst}{{C^\star}}

\newcommand{\X}{\mathcal{X}}

\newcommand{\cpst}[1]{c^{\star}_{#1}}

\newtheorem{fact}{Fact}

\usepackage{setspace}
\usepackage{mathtools}

\newcommand{\fair}[1]{\textup{fair-cost($#1$)}}
\newcommand{\costt}[1]{\textup{cost($#1$)}}

\usepackage{newfloat}
\DeclareFloatingEnvironment[fileext=lop]{Algorithm}

\newcommand{\RS}{{\tt Randomized-Subroutine}}

\newcommand{\rp}[2]{#1\strut^{\hspace{0.15mm}#2}}
\newcommand{\ep}[2]{#1\strut^{\,#2}}

\usepackage{scalerel}[2016/12/29]
\newcommand\z{\scaleobj{0.85}{z}}
\newcommand\zl{\scaleobj{1.15}{z}}

\allowdisplaybreaks

\graphicspath{{./Figures/}}

\begin{document}

\title{Tight FPT Approximation for Socially Fair Clustering}
%
%
\author{Dishant Goyal \and Ragesh Jaiswal}
%
%
%
 \institute{
Department of Computer Science and Engineering, \\
Indian Institute of Technology Delhi.\thanks{Email addresses: \email{\{Dishant.Goyal, rjaiswal\}@cse.iitd.ac.in}}
}
{\def\addcontentsline#1#2#3{}\maketitle}

\begin{abstract}
In this work, we study the {\em socially fair $k$-median/$k$-means problem}. We are given a set of points $P$ in a metric space $\X$ with a distance function $d(.,.)$. There are $\ell$ groups: $P_1,\dotsc,P_{\ell} \subseteq P$. 
We are also given a set $F$ of feasible centers in $\X$. 
The goal in the socially fair $k$-median problem is to find a set $C \subseteq F$ of $k$ centers that minimizes the maximum average cost over all the groups. That is, find $C$ that minimizes the objective function $\Phi(C,P) \equiv \max_{j} \Big\{ \sum_{x \in P_j} d(C,x)/|P_j| \Big\}$, where $d(C,x)$ is the distance of $x$ to the closest center in $C$. The socially fair $k$-means problem is defined similarly by using squared distances, i.e., $d^{2}(.,.)$ instead of $d(.,.)$. 
The current best approximation guarantee for both the problems is $O\left( \frac{\log \ell}{\log \log \ell} \right)$ due to Makarychev and Vakilian~\cite{fairness:2021_Socially_Makarychev}. 
In this work, we study the fixed parameter tractability of the problems with respect to parameter $k$. We design $(3+\veps)$ and $(9 + \veps)$ approximation algorithms for the socially fair $k$-median and $k$-means problems, respectively, in FPT (fixed parameter tractable) time $f(k,\veps) \cdot n^{O(1)}$, where  $f(k,\veps) = (k/\veps)^{{O}(k)}$ and $n = |P \cup F|$. Furthermore, we show that if Gap-ETH holds, then better approximation guarantees are not possible in FPT time.
\end{abstract}

\section{Introduction}\label{section:introduction}
\textit{Clustering} is a task of grouping the objects such that the objects within the same group are more similar to each other than the objects in the different groups. Clustering has been a well-studied topic. It has many mathematical formulations, heuristics, approximation algorithms, and a wide range of known applications (see~\cite{survey:survey_on_all_algorithms_2015} and~\cite{survey:akjain} for a brief survey). 
In recent years, the topic: \emph{fairness} in \emph{machine learning}, has gained considerable attention leading to its own dedicated conference: ACM FAccT (see~\cite{survey:2019_fairness_in_ML} and~\cite{survey:2020_fairness_Chouldechova} for the recent developments in this area). 
The main motivation is that in many human centric applications, the input data is biased towards a particular demographic group that may be based on age, gender, ethnicity, occupation, nationality, etc.
We do not want algorithms to discriminate among different groups due to biases in the dataset. 
In other words, we aim to design \emph{fair} algorithms for problems.

In the context of clustering, in particular the $k$-median/$k$-means/$k$-center clustering, various notions of fair clustering have recently been proposed (see for example:~\cite{fairness:2017_Chierichetti_NIPS,fairness:2019_Bera_NIPS,fairness:2019_Ahmadian,fairness:2019_Bercea_Schmidt_APPROX,fairness:2019_pranjal_kcenter,fairness:2019_XChen_Proportional,fairness:2020_Mahabadi_Vakilian}). Most of these notions are based on \emph{balanced} or \emph{proportionality} clustering. In other words, a clustering is said to be fair if in every cluster, a protected group (e.g. demographic group) occurs in an almost the same proportion as it does in the overall population. By the virtue of this, no group is over-represented or under-represented in any cluster. However, recently, Abbasi~\emph{et al.}~\cite{fairness:2021_Socially_Abbasi} demonstrated that ``balance'' based clustering is not desirable in applications where a cluster center represents an entire cluster. One such application is the placement of polling location for voting (see~\cite{fairness:2021_Socially_Abbasi} for details). In such applications, the quality of representation of a group is determined by the closeness of the group members to their cluster centers. Such cost representation is not captured by ``balance'' based clustering. Therefore, they introduced a new notion of the fair clustering where each group has an equitable cost representation in the clustering.
Informally, given a point set $P$ and $\ell$ groups: $P_{1},\dotsc,P_{\ell} \subseteq P$, the task is to cluster $P$ into $k$ clusters such that the maximum of the average costs of the groups is minimized.
In an independent work, Ghadiri~\emph{et al.}~\cite{fairness:2020_Socially_Ghadiri} used a similar notion that they called the ``socially fair'' clustering problem. Recently, Makarychev and Vakilian~\cite{fairness:2021_Socially_Makarychev} generalized the definition of the socially fair clustering problem using the weighted point set. The following is a formal definition of the problem as stated in~\cite{fairness:2021_Socially_Makarychev}.
\begin{definition}[Socially Fair Clustering]
We are given a set $P$ of points and set $F$ of feasible centers in a metric space $(\X,d)$. There are $\ell$ groups (possibly non-disjoint) of points $P_1, ..., P_{\ell} \subseteq P$ with weight function $w_j \colon P_j \to \mathbb{R}^{+}$ for each $j \in \{1,\dotsc,\ell\}$. Let $z$ be any real number $\geq 1$.
The unconstrained cost of a group $P_j$ with respect to a center set $C \subseteq F$ is defined as:
\[
\textup{cost}(C, P_j) \equiv \sum_{p \in P_j} d(C, p)^{\zl} \cdot w_j(p), \quad \textrm{where $d(C, p) \coloneqq \min_{c \in C} \big\{d(c, p)\big\}$}.
\]
In socially fair clustering, the goal is to pick a center set $C \subseteq F$ of size $k$ so as to minimize the objective function: $\max_{j \in [\ell]} \costt{C,P_j}$, which we call the fair cost:
\[
\fair{C,P} \equiv \max_{j \in [\ell]} \Big\{ \costt{C, P_j} \Big\}.
\]
\end{definition}

\noindent The case of averaging the cost of each group, i.e., $w_{j}: P_{j} \to 1/|P_j| $, was initially studied by Ghadiri~\emph{et al.}~\cite{fairness:2020_Socially_Ghadiri} and Abbasi~\emph{et al.}~\cite{fairness:2021_Socially_Abbasi} for which they gave an $O(\ell)$-approximation algorithm with polynomial running time. Furthermore, Abbasi~\emph{et al.}~\cite{fairness:2021_Socially_Abbasi} showed that the natural LP relaxation of the problem has an integrality gap of $\Omega(\ell)$. To overcome this barrier, Makarychev and Vakilian~\cite{fairness:2021_Socially_Makarychev} designed a strengthened LP and improved the approximation guarantee to $O\left( e^{O(z)} \cdot \frac{\log \ell}{\log \log \ell} \right)$ in polynomial time.  

Note that in the definition of the socially fair clustering, it is given that groups might not be disjoint. However, we can make the groups disjoint. If a point $p$ appears in multiple groups say $P_{j_{1}},\dotsc,P_{j_{t}}$, then we create $t$ copies of point $p$ such that its $i^{th}$ copy only belongs to $j_{i}^{th}$ group. Moreover, the weight of the $i^{th}$ copy is $w_{j_{i}}(p)$. The objective function does not change due to this modification. Therefore, from now on, we will assume all the groups to be disjoint. 

Now let us discuss some special cases of the problem. For $z = 1$ and $z = 2$, the problem is known as ``socially fair $k$-median'' and ``socially fair $k$-means'' problem,  respectively. On the other hand, if $z$ is arbitrary and $\ell = 1$, the problem is known as the ``$k$-service'' problem. Furthermore, in addition to $\ell = 1$, if $z = 1$ or $z = 2$, the problem becomes the classical (unconstrained) $k$-median/$k$-means problem, respectively. 

In this work, we study the fixed parameter tractability of the problem parameterized by $k$. It is known that the classical $k$-median and $k$-means problems when parameterized by $k$, are $\mathsf{W[2]}$-hard~\cite{fpt:2019_vincent}. Hence, it straightaway implies $\mathsf{W[2]}$-hardness of the socially fair clustering problem. Therefore, the problem does not admit an exact FPT algorithm unless $\mathsf{W[2]} = \mathsf{FPT}$. In this work, we design a $(\rp{3}{\z}+\veps)$-approximation algorithm for the problem, with FPT running time, parameterized by $k$. Furthermore, we show that this approximation guarantee is tight upto an $\veps$ additive factor.
Also, note that in the running time analysis of our algorithms, we ignore the dependence on $z$ since it is typically considered as constant. Formally, we state the main result as follows:

\begin{theorem} [Main Theorem]\label{theorem:main_theorem}
Let $z \geq 1$ and $0 \leq \veps \leq 1$. Let $\mathcal{I} = (\X,P,P_1,\dotsc,P_{\ell},w_1,\dotsc,w_{\ell},$ $F,d,k,z)$ be any instance of the socially fair clustering problem. Then, there is a randomized algorithm that outputs a $(\rp{3}{\z}+\veps)$-approximate solution to $\mathcal{I}$ with probability at least $1-1/n$. The running time of the algorithm is $\left( k/\veps \right) ^{O(k)} \cdot n^{O(1)}$, which is FPT in $k$.
\end{theorem}

\noindent The following are two immediate corollaries of the above theorem.

\begin{corollary}[$k$-median]
For the socially fair $k$-median problem, there is a randomized $(3+\veps)$-approximation algorithm with FPT running time of $(k/\veps)^{O(k )} \cdot n^{O(1)}$ that succeeds with probabaility at least $1-1/n$.
\end{corollary}

\begin{corollary}[$k$-means]
For the socially fair $k$-means problem, there is a randomized $(9+\veps)$-approximation algorithm with FPT running time of $(k/\veps)^{O(k )} \cdot n^{O(1)}$ that succeeds with probabaility at least $1-1/n$. 
\end{corollary}

In Section~\ref{section:lower_bound}, we establish FPT hardness of approximation results for the problem that follow from the known hardness results of the unconstrained clustering problems. The following are two main results:

\begin{theorem} [FPT Hardness for Parameters: $\ell$ and $k$]\label{theorem:hardness_k_ell}
For any constant $\veps>0$ and functions: $g \colon \mathbb{R}^{+} \to \mathbb{R}^{+}$ and $f  \colon \mathbb{R}^{+} \to \mathbb{R}^{+}$, the socially fair $k$-median and $k$-means problems can not be approximated to factors $(1+2/e-\veps)$ and $(1+8/e-\veps)$, respectively, in time $g(k) \cdot n^{f(\ell) \cdot o(k) }$, assuming Gap-ETH.
\end{theorem}

\begin{theorem}[FPT Hardness for Parameter $k$]\label{theorem:hardness_k}
For any $z \geq 0$, $\veps>0$, and function $g \colon \mathbb{R}^{+} \to \mathbb{R}^{+}$, the socially fair clustering problem can not be approximated to factor $(\rp{3}{\z} - \veps)$ in time $g(k) \cdot n^{o(k)}$, assuming Gap-ETH.
\end{theorem}

\noindent  This completes the summary of our results. Note that Theorems~\ref{theorem:main_theorem} and~\ref{theorem:hardness_k} give tight approximation bounds for the socially fair clustering problem when parameterized by $k$. Thus, it settles the complexity of the problem when parameterized by $k$. Next, we compare our work with the previous related works.

\section{Related Work}
The previous works of Abbasi~\emph{et al.}~\cite{fairness:2021_Socially_Abbasi}, and Makarychev and Vakilian~\cite{fairness:2021_Socially_Makarychev} were based on the LP relaxation and rounding techniques. They gave $O(\ell)$ and $O\left( e^{O(z)} \cdot \frac{\log \ell}{\log \log \ell} \right)$ approximation guarantees, respectively, for the socially fair clustering problem. On the other hand, Ghadiri~\emph{et al.}~\cite{fairness:2020_Socially_Ghadiri} designed a socially fair $k$-means algorithm with performance guarantees similar to the Lloyd's heuristics~\cite{kmeans:lloyds} (popularly known as the \emph{$k$-means algorithm}). In contrast, we use a simple bi-criteria approximation algorithm as a subroutine to obtain FPT time constant factor approximation algorithm for the problem. 

Recently, Bandyapadhyay~\emph{et al.}~\cite{fairness:2021_constrained_Bandyapadhyay} gave an FPT time constant factor approximation algorithm for a variant of ``balance'' based fair clustering problem. This variant was first studied by Chierichetti \emph{et al.}~\cite{fairness:2017_Chierichetti_NIPS} and later generalized by Bera~\emph{et al.}~\cite{fairness:2019_Bera_NIPS}. According to this variant, a clustering is said to be fair if within each cluster, the fraction of points that belongs to the $j^{th}$ group is at least $\beta_{j}$ and at most $\alpha_{j}$, for some constants $0\leq\alpha_{j},\beta_{j}\leq1$. 
This results in fair representation of every group within each cluster. 
It turns out that this variant falls under a broad class of the \emph{constrained} $k$-median/$k$-means problem~\cite{fairness:2021_constrained_Bandyapadhyay,constrained:2015_Ding_and_Xu}. 
Informally, the constrained $k$-median/$k$-means problem is a class of clustering problems where a set of constraints can be imposed on the clusters in addition to optimising the $k$-means/$k$-median cost.
Various other problems like:~\emph{uniform capacitated} $k$-median/$k$-means problem~\cite{capacitated:kmedian_2017_Li_uniform}, \emph{outlier} $k$-median/$k$-means problem~\cite{outlier:kmeans_2018_Ravishankar}, \emph{fault-tolerant} $k$-median/$k$-means problem~\cite{fault:kmedian_2014_non_uniform_haji_li_SODA}, etc., fall in this category.

Moreover, it is known that if any constrained $k$-median/$k$-means problem admits an FPT time \emph{partition algorithm}, then it also admits an FPT time constant factor approximation algorithm, in general metric spaces~\cite{constrained:2020_GJK_FPT}. We refer the reader to~\cite{constrained:2015_Ding_and_Xu,constrained:2016_bjk,constrained:2020_GJK_FPT} for the definitions of constrained clustering, partition algorithm, and other related constrained clustering examples.
We skip these details since they are not immediately relevant to our discussion; 
we just wanted to convey the high level idea. 
Since Chierichetti~\emph{et al.}'s~\cite{fairness:2017_Chierichetti_NIPS} definition of the fair clustering fits the constrained clustering framework, it is tempting to check if the socially fair clustering problem fits the constrained clustering framework. Unfortunately, the objective function of socially fair clustering differs from the classical $k$-median and $k$-means objectives. 
Therefore, the problem cannot be treated as a constrained clustering problem.
However, we note that the cost function for each group $P_{j}$ is exactly the same as the $k$-median/$k$-means objective. 
We use this fact to design a polynomial time bi-criteria approximation algorithm for the problem. Then, we convert the bi-criteria approximation algorithm to a constant factor approximation algorithm in FPT time. We will formally define the bi-criteria approximation algorithm in Section~\ref{subsection:1}.

Another way of approaching this problem is to obtain a \emph{strong coreset} for the socially fair clustering instance. The coreset can be easily obtained by computing the coresets for each group $P_j$ individually. For the coreset definition and its construction, see the work of Ke Chen~\cite{coreset:Chen09}, or Cohen-Addad~\cite{coreset:vincent_STOC_21}. After obtaining a coreset of the point set $P$, one can employ the techniques of Cohen-Addad~\emph{et al.}~\cite{fpt:2019_vincent}, and Cohen-Addad and Li~\cite{capacitated:FPT_2019_vincent} to obtain a constant factor approximation for the problem in FPT time. 
The main idea is to try all possible $k$ combination of points in the coreset and choose the centers in $F$ that are closest to those points. This gives ${|S| \choose k}$ distinct center sets, where $|S|$ is the number of points in the coreset. It can then be shown that the center set that gives the least clustering cost is a $(\rp{3}{\z}+\veps)$ approximation for the problem. 
For details, see Section 2.2 of~\cite{capacitated:FPT_2019_vincent} in the context of $k$-median and $k$-means objectives.
However, there are an issue with this technique when we deal with the socially fair clustering objective. The issue is that the coreset size would have an $\ell$ term, where $\ell$ is the number of groups. Therefore, the running time would have a multiplicative factor of $\ell^{\,k}$ which makes the algorithm not be FPT in $k$. Also note that $\ell$ can be as large as $\Omega(n)$. In this work, without using coreset techniques, we design a $(\rp{3}{\z}+\veps)$ approximation algorithm for the problem. Moreover, the running time of the algorithm is FPT in $k$. In the following section, we mention some notations and facts that we use frequently in this paper.

\section{Notations and Facts} \label{section:notations}
Let $\mathcal{I} = (\X,P,P_1,\dotsc,P_{\ell},w_1,\dotsc,w_{\ell},F,d,k,z)$ be an instance of the socially fair clustering problem. For a weighted set $S \subseteq P$ with weight function $w \colon S \to \mathbb{R}^{+}$ and a center set $C \subseteq F$, we denote the \emph{unconstrained} clustering cost of $S$ with respect to $C$ by $\costt{C,S}$, i.e, \[\costt{C,S} \equiv \sum_{p \in S} d(C,p)^{\zl} \cdot w(p), \textup{ where } d(C,p) = \min_{c \in C} \{ d(c,p) \}.\] 
For simplicity, when $S = \{p\}$, we use the notation $\costt{C,p}$ instead of $\costt{C,\{p\}}$. Similarly, when $C = \{c\}$, we use the notation $\costt{c,S}$ instead of $\costt{\{ c\},S}$. In the remaining discussion, we will refer to the \emph{unconstrained clustering cost} simply as \emph{clustering cost}.

\noindent We denote the \emph{fair clustering} cost of $P$ with respect to a center set $C$ by $$\fair{C,P} \equiv \max_{j \in [\ell]} \Big\{ \costt{C,P_j} \Big\}.$$
Moreover, we denote the optimal fair clustering cost of $P$ by $OPT$ and optimal fair center set by $\Cst = \{ \cpst{1}, \dotsc, \cpst{k}\}$, i.e., $\fair{\Cst,P} = OPT$. We also use the notation $[t]$ to denote a set $\{1,\dotsc,t\}$ for any integer $t \geq 1$.
We also use the following inequality in our proofs. The inequality is a generalization of the triangle inequality and easily follows from the \emph{power-mean} inequality.
\begin{fact}[Approximate Triangle Inequality]\label{fact:1}
For any $z\geq 1$, and any four points: $q,r,s,t \in \X$, $d(q,t)^{\zl} \leq (d(q,r) + d(r,s) + d(s,t))^{\zl} \leq \rp{3}{\z-1} \cdot \left(d(q,r)^{\zl} + d(r,s)^{\zl} + d(s,t)^{\zl} \right)$.
\end{fact}

\section{FPT Approximation}\label{section:fpt_approximation}

In this section, we design a $(\rp{3}{\z}+\veps)$-approximation algorithm for the socially fair clustering problem with FPT time of $\left( k/\veps \right) ^{O(k)} \cdot n^{O(1)}$. 
The algorithm turns out to be surprisingly simple. 
Our algorithm consists of the following two parts:
\begin{enumerate}
    \item A polynomial time $\left( O\left( (k/\veps^2) \cdot \ln^{2} n \right),1+\veps \right)$ bi-criteria approximation algorithm for the socially fair clustering problem.
    \item We use the above bi-criteria algorithm to obtain a center set $C \subseteq F$ of size $O((k/\veps^2) \cdot \ln^{2}n)$. We then show that there exists a $k$-sized subset $S \subset C$ that gives $(\rp{3}{\z}+\veps)$ approximation. Note that since one needs to try all possible $k$-sized subsets of $C$, the overall running time of the algorithm has a multiplicative factor of $O\binom{|C|}{k}$. This results in an FPT algorithm.
\end{enumerate}

\noindent We discuss the above two parts in Sections~\ref{subsection:1} and~\ref{subsection:2}.

\subsection{Bi-criteria Approximation}\label{subsection:1}

We start with the definition of $(\alpha,\beta)$ bi-criteria approximation algorithm:
\begin{definition} [Bi-criteria Approximation]
An algorithm is said to be $(\alpha,\beta)$ bi-criteria approximation for the problem if it outputs a set $C$ of $\beta k$ centers with fair clustering cost at most $\alpha$ times the optimal fair clustering cost with $k$ centers, i.e.,
\[
\fair{C,P} \leq \alpha \cdot \min_{|C'| = k \textrm{ and } C' \subseteq F} \Big\{ \, \fair{C',P} \, \Big\} = \alpha \cdot OPT
\]
\end{definition}

Note that for the \emph{unconstrained} clustering problem, there exists a randomized $\left( O(k \ln (1/\veps), 1+\veps)\right)$ bi-criteria approximation algorithm due to Neal Young~\cite{SE:neal_young}. 
We extend that algorithm to the socially fair clustering problem and obtain a randomized $\left( O\left( (k/\veps^2) \cdot \ln^{2} n \right),1+\veps \right)$ bi-criteria approximation algorithm. 
Formally, we state the result as follows:

\begin{theorem} [Fair Bi-Criteria Approximation]\label{theorem:bi-criteria}
Let $\mathcal{I} = (\X,P,P_1,\dotsc,P_{\ell},w_1,\dotsc,w_{\ell},F,d,k,z)$ be any instance of the socially fair clustering problem. Then, there exists a polynomial time algorithm that with probability at least $(1-1/n)$ outputs a center set $C \subseteq F$ of size $O\left( (k/\veps^2) \cdot \ln^{2} n \right)$ that is a $(1+\veps)$-approximation to the optimal fair clustering cost of $\mathcal{I}$ with $k$ centers. That is, $\fair{C,P} \leq (1+\veps) \cdot OPT$.
\end{theorem}

The above theorem follows from the next two lemmas. 
Since the proof of these lemmas follows from known techniques, we defer the detailed proof to the Appendix. 
Here, we discuss the techniques involved.

\begin{lemma}\label{lemma:bicriteria}
There is a polynomial time randomized algorithm \RS\ that outputs a center set $C'$ such that for every group $P_j \in \{ P_1,\dotsc,P_{\ell} \}$, the expected clustering cost of $P_j$ with respect to $C'$ is at most $(1+\veps/2)$ times the optimal fair clustering cost of instance $\mathcal{I}$. That is, for all $j$,
$$
\mathbb{E} \left[ \, \costt{C',P_j} \, \right] \leq \left( 1+\frac{\veps}{2} \right) \cdot OPT
$$
\end{lemma}

The above lemma follows from a modification of the known bi-criteria approximation algorithm for the unconstrained clustering problem \cite{SE:neal_young} which in turn follows from an LP-rounding technique with respect to the most natural linear programming formulation of the problem. 
We give the outline of the Linear Programming (LP) relaxation and the rounding procedure while deferring the analysis to the Appendix.
We start with the natural LP-relaxation for the problem:
\begin{framed}
\[
    \hspace{-11.7cm} \textrm{minimize} \quad  \gamma 
\]
\vspace{-7.5mm}
\begin{align*}
    \textrm{subject to} \quad &\sum_{f \in F} y_{f} = k &&\\
    &\sum_{f \in F} x_{f,p} = 1 \quad &&\textrm{for every point $p \in P$}\\
    &\sum_{p \in P_j} \sum_{f \in F} x_{f,p} \cdot  d(f,p)^{\zl} \cdot w_j(p) \leq \gamma  \quad &&\textrm{for every $P_j \in \{P_1,\dotsc,P_{\ell} \}$} \\
    & x_{f,p} \leq y_{f} \quad &&\textrm{for every $f \in F$ and $p \in P$} \\
    &y_f, x_{f,p} \geq 0 \quad &&\textrm{for every $f \in F$ and $p \in P$} 
\end{align*}
\vspace{-5.0mm}
\end{framed}
\noindent Here, $y_f$ is a variable that denote the fraction of a center $f$ picked in the solution. The variable $x_{f,p}$ denote the fraction of point $p$ assigned to center $f$. The variable $\gamma$ denote the fair clustering cost of a feasible fractional solution. We solve the above linear program to obtain the fractional optimal solution: $y_f^{*}$,  $x_{f,p}^{*}$, and $\gamma^{*}$. Since it is a relaxation to the original problem, $\gamma^{*} \leq OPT$. For simplicity, we use the notations: $y_f$, $x_{f,p}$, $\gamma$ for $y_f^{*}$, $x_{f,p}^{*}$, and $\gamma^{*}$, respectively.

\begin{Algorithm}[h]
\begin{framed}
\RS~($\mathcal{I}$, $y_f$'s, $x_{f,p}$'s, $\veps$) \vspace{2mm}\\
\hspace*{0.2in} {\bf Inputs}: Socially fair clustering instance $\mathcal{I}$, fractional optimal solution: $y_f$'s and $x_{f,p}$'s of \\
\hspace*{0.78in} the relaxed LP, and accuracy $\veps \leq 1$. \\
\hspace*{0.2in} {\bf Output}: A center set $C' \subseteq F$ of size $O(k \ln (n/\veps))$ such that \\
\hspace*{0.89in}$\mathbb{E} \left[ \costt{C',P_j} \, \right] \leq (1+\veps/2) \cdot OPT $ for every group $P_j \in \{P_1,\dotsc,P_{\ell} \}$ \vspace*{2mm}\\
\hspace*{0.0mm} (1) \ \ \ $C' \gets \emptyset$ \quad \emph{ (center set)}  \\
\hspace*{0.0mm} (2) \ \ \ $P_{u} \gets P$ \quad \emph{(set of unassigned points)}  \\
\hspace*{0.0mm} (3) \ \ \ Repeat $k \ln (2 c \cdot n/\veps)$ times for some constant $c$\,:  \\
\hspace*{0.0mm} (4) \hspace*{0.3in} \ \ \ Sample a center $f \in F$ with probability $y_f/k$. Let $f^{*}$ be the sampled center. \\
\hspace*{0.0mm} (5) \hspace*{0.3in} \ \ \ $C' \gets C' \cup \{f^{*}\}$\\
\hspace*{0.0mm} (6) \hspace*{0.3in} \ \ \ For each point $p \in P_{u}$: \\
\hspace*{0.0mm} (7) \hspace*{0.3in} \hspace*{0.3in} \ \ \ Assign $p$ to $f^{*}$ with probability $x_{p,f^{*}}/y_{f^{*}}$  \\
\hspace*{0.0mm} (8) \hspace*{0.3in} \hspace*{0.3in} \ \ \ If $p$ assigned to $f^{*}$, then $P_u \gets P_u \setminus \{p\}$  \\ 
\hspace*{0.0mm} (9) \ \ \ Run $O(\ell)$-approximation algorithm for socially fair clustering problem on $P$.\\
\hspace*{0.0mm} \hspace*{0.35in} Let $C_u \subseteq F$ be the obtained center set. Assign the points in $P_u$ to $C_u$.\\
\hspace*{0.0mm} (10) \ \ $P_u \gets \emptyset$ \\
\hspace*{0.0mm} (11) \ \ $C' \gets C' \cup C_{u}$ \\
\hspace*{0.0mm} (12) \ \  return($C'$)
\end{framed}
\vspace*{-4mm}
\caption{A rounding procedure used as a subroutine for bi-criteria approximation.}
\label{algorithm:RS}
\end{Algorithm}

\noindent
The randomized subroutine is described in Algorithm~\ref{algorithm:RS}. The algorithm takes input the fractional optimal solution to the relaxed LP of the socially fair clustering problem. In Line $(1)$, the algorithm initializes a center set $C'$ as empty. In Line $(2)$, the algorithm initializes a set $P_{u}$ that denotes the set of unassigned points. Initially, no point is assigned to any center; therefore $P_{u}$ is initialized to $P$. Then, the algorithm proceeds in two phases:
\begin{list}{$\bullet$}{} 
    \item Phase $1$ constitutes Lines $(3)-(8)$ of the algorithm. In this phase, the algorithm samples a center from $F$ with probability distribution defined by $y_f/k$. Note that sum of probabilities over all $f \in F$, is $1$ due to constraint $(1)$ of the relaxed LP. Therefore, a center is always selected, and it added to $C'$. Suppose the selected center is $f^{*}$. Then, for each point $p \in P_{u}$, the algorithm independently assigns $p$ to $f^{*}$ with probability $x_{f^{*},p}/y_{f^{*}}$. The algorithm removes the points from $P_{u}$ that are assigned to $f^{*}$. The algorithm repeats this procedure $k \ln (2c \cdot n/\veps)$ times. Here, $c$ is a constant whose value will be defined later during the analysis of the algorithm. \\
    
    \item Phase $2$ constitutes line $(9)$ of the algorithm. In this phase, the algorithm runs an $O(\ell)$-approximation algorithm $\mathcal{A}$ for the socially fair clustering problem on the entire point set $P$. It is easy to design such an algorithm. In fact, any $O(1)$-approximation algorithm to the \emph{unconstrained} clustering problem is also an $O(\ell)$-approximation  algorithm to the socially fair clustering problem. For the sake of completeness, we prove this reduction in Appendix~\ref{appendix:ell_approx}. Let $C_{u}$ be the center set output by algorithm $\mathcal{A}$. The remaining points in $P_u$ are assigned to their closest centers in $C_{u}$. After this phase, all the points are assigned. Lastly, the algorithm returns all the centers selected in Phase $1$ and Phase $2$.
\end{list}
The details of the proof of Lemma~\ref{lemma:bicriteria} which mainly involves analysis of \RS\,is given in Appendix~~\ref{appendix:randomized_sub}.
We now apply standard probability amplification method to bound the fair-cost. We give the proof in Appendix~\ref{sec:amplification}.

\begin{lemma}\label{lemma:amplification}
Suppose \RS \, is repeated $r = \frac{8 \ln n}{\veps}$ times, independently. Let $C'_{1},\dotsc,C'_{r}$ be the obtained center sets for each call to the algorithm. Then, the center set $C \coloneqq C_{1}' \cup \dotsc \cup C_{r}'$ is a $(1+\veps)$ approximation to the optimal fair clustering cost of $P$, i.e., $\fair{C,P} \leq (1+\veps) \cdot OPT$, with probability at least $1-1/n$.
\end{lemma}

\subsection{Conversion: Bi-criteria to FPT Approximation} \label{subsection:2}

In this subsection, we convert the $\left( O\left( (k/\veps^2) \cdot \ln^{2} n \right),1+\veps \right)$ bi-criteria approximation algorithm to $(\rp{3}{\z}+\veps)$-approximation algorithm in FPT time.

\begin{lemma} \label{lemma:main}
Let $C = \{c_1,\dotsc,c_{\beta k} \} \subseteq F$ be any $(\alpha,\beta)$-approximate solution to the socially fair clustering instance $\mathcal{I} = (\X,P,P_1,\dotsc,P_{\ell},w_1,\dotsc,w_{\ell},F,d,k,z)$. Then, there exists a $k$ sized subset $C'$ of $C$ that is a $(\rp{3}{\z-1} \cdot (\alpha+2))$-approximate solution to $\mathcal{I}$. Moreover, given $C$, the center set $C'$ can be obtained in time $O((e\beta)^{\,k} \cdot nk)$.
\end{lemma}

\begin{proof}
Let $C^{*} = \{c_{1}^{*},\dotsc,c_{k}^{*} \} \subseteq F$ be an optimal center set of $\mathcal{I}$. This set induces a Voronoi partitioning in each of the groups. 
We denote this partitioning using the notation $\mathbb{P}_{j} = \{P_j^1,P_j^2, \dotsc,P_j^k\}$ for the $j^{\textrm{th}}$ group. 
That is, $P_j^i$ are the set of those points in group $P_j$ for which the center $\cpst{i}$ is the closest.
For any point $x \in P \cup F$, let $f(x)$ denote the point in $C$ that is closest to $x$. That is, $f(x) \coloneqq \arg\min_{c \in C} \big\{ d(c,x) \big\}$. We define a new center set $C' \coloneqq \{f(c_{1}^{*}),\dotsc,f(c_{k}^{*}) \} \subseteq C$. We show that $C'$ is a $(\rp{3}{\z-1} \cdot (\alpha+2))$-approximate solution to $\mathcal{I}$. The proof follows from the following sequence of inequalities:


\begin{align*}
\fair{C',P} &= \max_{j \in [\ell]} {\left[ \costt{C',P_j}\right]} &&\\
&\leq \max_{j \in [\ell]}{\left[ \sum_{i=1}^{k} \costt{f(c_{i}^{*}), P_j^i}\right]} &&\\
&= \max_{j \in [\ell]}{\left[ \sum_{i=1}^{k} \sum_{x \in P_{j}^{i}} w_{j}(x) \cdot d (f(c_{i}^{*}),x)^{\zl}  \right]} &&\\
&\leq \max_{j \in [\ell]}{\left[ \sum_{i=1}^{k} \sum_{x \in P_{j}^{i}} w_{j}(x) \cdot \left( d(x,\cpst{i}) + d(\cpst{i},f(c_{i}^{*})) \right)^{\zl} \right]}, \quad && \textrm{(using triangle inequality)}\\
&\leq \max_{j \in [\ell]}{\left[ \sum_{i=1}^{k} \sum_{x \in P_{j}^{i}} w_{j}(x) \cdot \left( d(x,\cpst{i}) + d(\cpst{i},f(x)) \right)^{\zl} \right]}, \quad && \textrm{(from definition of $f(c_{i}^{*})$)}\\
&\leq \max_{j \in [\ell]}{\left[ \sum_{i=1}^{k} \sum_{x \in P_{j}^{i}} w_{j}(x) \cdot \left( d(x,\cpst{i}) + d(x,\cpst{i}) + d(x,f(x)) \right)^{\zl} \right]}, \quad && \textrm{(using triangle inequality)}\\
\end{align*}
\vspace{-13mm}
\begin{align*}
\hspace{25mm} &\leq \max_{j \in [\ell]}{\left[ \rp{3}{\z-1} \cdot \sum_{i=1}^{k} \sum_{x \in P_{j}^{i}} w_{j}(x) \cdot \left( d(x,\cpst{i})^{\zl} + d(x,\cpst{i})^{\zl} + d(x,f(x))^{\zl} \right)  \right]}, \quad \textrm{(using Fact~\ref{fact:1})}\\
&\leq \max_{j \in [\ell]}{\left[ \rp{3}{\z-1} \cdot  2 \cdot \sum_{i=1}^{k} \sum_{x \in P_{j}^{i}} w_{j}(x) \cdot  d(x,\cpst{i})^{\zl}  \right]} + \max_{j \in [\ell]}{\left[ \rp{3}{\z-1} \cdot \sum_{i=1}^{k} \sum_{x \in P_{j}^{i}} w_{j}(x) \cdot d(x,f(x))^{\zl}   \right]}\\
\end{align*}
\vspace{-13mm}
\begin{align*}
\hspace{17mm} &=  \max_{j \in [\ell]}{\left[ 2 \cdot \rp{3}{\z-1} \cdot \sum_{i=1}^{k} \costt{\cpst{i}, P_j^i}\right]} +  \max_{j \in [\ell]}{\left[ \rp{3}{\z-1} \cdot \costt{C,P_{j}} \right]}&&\\
&\leq  2 \cdot \rp{3}{\z-1} \cdot OPT +  \max_{j \in [\ell]}{\left[ \rp{3}{\z-1} \cdot \costt{C,P_{j}} \right]}&&\\
&\leq 2 \cdot \rp{3}{\z-1} \cdot OPT + \rp{3}{\z-1} \cdot \alpha \cdot OPT,  \quad \quad \textrm{($\because C$ is an $\alpha$-approximate solution)} &&\\
&= \rp{3}{\z-1} \cdot (\alpha+2) \cdot OPT.
\end{align*}
This proves that $C'$ is a $(\rp{3}{\z-1} \cdot (\alpha+2))$-approximate solution to $\mathcal{I}$. 

Now, we find the center set $C'$ using $C$. Since, we do not know $C^{*}$, we can not directly find $C'$. Therefore, we take all possible $k$ sized subsets of $C$ and compute the fair clustering cost for each of them. We output that center set that gives the least fair clustering cost.
There are at most ${\beta k \choose k} \leq (e\beta)^{ \, k}$ possibilities
\footnote{Here, we use a well known inequality that ${n \choose k} \leq (\frac{e  \cdot n}{k})^{ \, k}$.} of $C'$. 
And, for each such center set, the fair clustering cost can be computed in $O(nk)$ time using the Voronoi partitioning algorithm. Therefore, the overall running time is $O((e\beta)^{k} \cdot nk)$. This completes the proof of the lemma. \qed
\end{proof}

In the above lemma, we substitute the bi-criteria approximation algorithm that we designed in the previous subsection. It had $\alpha = 1+\veps$ and $\beta = O\left( (k/\veps^2) \cdot \ln^{2} n \right)$. Moreover, we set $\veps = \veps'/\rp{3}{\z-1}$ for some constant $\veps'\leq 1$. Then, the above lemma gives the following main result for the fair clustering problem:

\begin{corollary} [Main Result]
Let $z \geq 1$ and $\veps'$ be any constant $\leq 1$. Let $\mathcal{I} = (\X,P,P_1,\dotsc,P_{\ell},w_1,$ $\dotsc,w_{\ell},F,d,k,z)$ be any instance of the socially fair clustering problem. Then, there is an algorithm that outputs a $(\rp{3}{\z}+\veps')$-approximate solution for $\mathcal{I}$ with probability at least $1-1/n$. The running time of the algorithm is $\left( k/\veps' \right) ^{O(k)} \cdot n^{O(1)}$, which is FPT in $k$.
\end{corollary}
\begin{proof}
It is easy to see that the approximation guarantee of the algorithm is $(\rp{3}{\z} + \veps')$ since $\alpha = 1+\veps$ and $\veps = \veps'/\rp{3}{\z-1}$. 
Note that the bi-criteria approximation algorithm has polynomial running time. Furthermore, converting the bi-criteria approximation algorithm to FPT approximation algorithm requires $O((e (k/\veps^2) \cdot \ln^{2} n)^{k} \cdot nk)$ time since $\beta =O\left( (k/\veps^2) \cdot \ln^{2} n \right)$. Using a standard inequality that $(\ln n)^{k} = k^{O(k)} \cdot n$, we get total running time of $\left( k/\veps' \right) ^{O(k)} \cdot n^{O(1)}$, which is FPT in $k$. Hence proved.\qed
\end{proof}

\section{FPT Lower Bounds}\label{section:lower_bound}
In this section, we establish the FPT hardness of approximation results for the socially fair clustering problem. For this, we use the hardness results of the unconstrained $k$-median and $k$-means problems. Firstly, note the following result from~\cite{fpt:2020_Pasin_SODA}. 

\begin{theorem}[Corollary 3 of~\cite{fpt:2020_Pasin_SODA}]\label{theorem:hardness_kmedian}
For any constant $\veps>0$ and any function $g \colon \mathbb{R}^{+} \to \mathbb{R}^{+}$, the $k$-median and $k$-means problems can not be approximated to factors $(1+2/e-\veps)$ and $(1+8/e-\veps)$, respectively, in time $g(k) \cdot
n^{o(k)}$, assuming Gap-ETH. 
\end{theorem}

\noindent It is easy to see that for $\ell = 1$, the socially fair $k$-median/$k$-means problem is equivalent to the unconstrained $k$-median/$k$-means problem. This gives the following hardness result for the socially fair $k$-median/$k$-means problem.
\begin{theorem} [FPT Hardness for Parameters: $\ell$ and $k$]\label{theorem:hardness_parameter_k_l}
For any constant $\veps>0$, and functions: $g \colon \mathbb{R}^{+} \to \mathbb{R}^{+}$ and $f  \colon \mathbb{R}^{+} \to \mathbb{R}^{+}$, the socially fair $k$-median and $k$-means problems can not be approximated to factors $(1+2/e-\veps)$ and $(1+8/e-\veps)$, respectively, in time $g(k) \cdot n^{f(\ell) \cdot o(k) }$, assuming Gap-ETH.
\end{theorem}
\begin{proof}
For the sake of contradiction, assume that there exists a constant $\veps>0$, and functions: $g \colon \mathbb{R}^{+} \to \mathbb{R}^{+}$ and $f \colon \mathbb{R}^{+} \to \mathbb{R}^{+}$ such that the socially fair $k$-median and $k$-means problems can be approximated to factors $(1+2/e-\veps)$ and $(1+8/e-\veps)$, respectively, in time $g(k) \cdot n^{f(\ell) \cdot o(k) }$. Then, for $\ell = 1$, it implies that the $k$-median and $k$-means problems can be approximated to factors $(1+2/e-\veps)$ and $(1+8/e-\veps)$, respectively, in time $g(k) \cdot n^{o(k)}$. This contradicts Theorem~\ref{theorem:hardness_kmedian}. Hence proved. 
\end{proof}

The above hardness result assumed the parametrization by $k$ and $\ell$. Now we show stronger hardness result for the problem when it is parameterized by $k$ alone. We show this using a reduction from the \emph{$k$-supplier} problem. The $k$-supplier problem is defined as follows:
\begin{definition}[$k$-Supplier]
Let $z$ be any positive real number and $k$ be any positive integer. Given a set $P$ of points and set $F$ of feasible centers in a metric space $(\X,d)$, find a set $C \subseteq F$ of $k$ centers that minimizes the objective function $\Phi(C,P)$ defined as follows:
\[
\Phi(C,P) \equiv  \max_{x \in P} \Big\{ d(C,x)^{\zl} \Big\} , \quad \textrm{where} \quad d(C,x) = \min_{c \in C} \{ d(c,x) \}
\]
\end{definition}

\noindent Hochbaum and Shmoys~\cite{Supplier:1986_Shmoys} showed that for $z = 1$, the $k$-supplier problem is $\mathsf{NP}$-hard to approximate to any factor smaller than $3$. The proof follows from the reduction from the \emph{hitting set problem} (see Theorem 6 of~\cite{Supplier:1986_Shmoys}). A similar reduction is possible from the \emph{set coverage problem}. For the sake of completeness, we describe the reduction in Appendix~\ref{appendix:reduction}. The reduction gives the following FPT hardness of approximation result for the $k$-supplier problem.

\begin{theorem}\label{theorem:hardness_supplier}
For any $\veps>0$ and any function $g \colon \mathbb{R}^{+} \to \mathbb{R}^{+}$, the $k$-supplier problem can not be approximated to factor $(\rp{3}{\z} - \veps)$ in time $g(k) \cdot n^{o(k)}$, assuming Gap-ETH.
\end{theorem}

\noindent Using this, we show the following hardness result for the socially fair clustering problem.
\begin{theorem}[FPT Hardness for Parameter $k$]\label{theorem:hardness_parameter_k}
For any $z \geq 0$, $\veps>0$, and any function $g \colon \mathbb{R}^{+} \to \mathbb{R}^{+}$, the socially fair clustering problem can not be approximated to factor $(\rp{3}{\z} - \veps)$ in time $g(k) \cdot n^{o(k)}$, assuming Gap-ETH.
\end{theorem}
\begin{proof}
We prove this result by showing that the $k$-supplier problem is a special case of the socially fair clustering problem.
Let $\mathcal{I} = (P,P_1,\dotsc,P_{\ell},w_1,\dotsc,w_{\ell},F,d,k,z)$ be an instance of the socially fair clustering problem defined in the following manner. The number of groups is the same as the number of points, i.e., $\ell = |P|$. For every point $x \in P$, we define a singleton group $P_x$ as  $P_{x} \coloneqq \{x\}$. Let the weight function be defined as $w_{x} \colon P_{x} \to 1$ for every $x \in P$, i.e., each point carries a unit weight. Then, for any center set $C \subseteq F$, the fair clustering cost of $P$ is:
\[
\fair{C} = \max_{x \in P} \Big\{ \costt{C,P_{x}} \Big\}= \max_{x \in P} \Big\{ d(C,x)^{\zl} \Big\}
\]
Recall that $\max_{x \in P} \Big\{ d(C,x)^{\zl} \Big\}  \equiv \Phi(C,P)$ is the $k$-supplier cost of the instance $(P,F,d,k)$. It means that the fair cost of instance $\mathcal{I}$ is the same as the $k$-supplier cost of the instance $(P,F,d,k)$. 
Therefore, the $k$-supplier problem is a special case of the socially fair clustering problem. Therefore, the hardness result stated in Theorem~\ref{theorem:hardness_supplier} also holds for the socially fair clustering problem. Hence proved.
\end{proof}

\noindent The following are the two corollaries that immediately follow from the above theorem.
\begin{corollary}[$k$-median]
For any $\veps>0$ and any computable function $g(k)$, the socially fair $k$-median problem can not be approximated to factor $(3-\veps)$ in time $g(k) \cdot n^{o(k)}$, assuming Gap-ETH.
\end{corollary}

\begin{corollary}[$k$-means]
For any $\veps>0$ and any computable function $g(k)$, the socially fair $k$-means problem can not be approximated to factor $(9-\veps)$ in time $g(k) \cdot n^{o(k)}$, assuming Gap-ETH. 
\end{corollary}
\section{Conclusion}
We designed a constant factor approximation algorithm for the socially fair $k$-median/$k$-means problem in FPT time. In addition to it, we gave tight FPT hardness of approximation bound using the observation that the $k$-supplier problem is a special case of the socially fair clustering problem. This settles the complexity of the problem when parameterized by $k$.
A natural open question is to obtain better approximation guarantees parameterized by both $k$ and $\ell$, or $\ell$ alone.
\section*{Acknowledgement}
Thanks to Karthik C.S. for pointing us to the work of Pasin Manurangsi~\cite{fpt:2020_Pasin_SODA}.

\addcontentsline{toc}{section}{References}
\bibliographystyle{alpha}
\bibliography{references}

\appendix

\section{Randomized Subroutine (Proof of Lemma~\ref{lemma:bicriteria})}\label{appendix:randomized_sub}

Here, we analyse the rounding procedure \RS.
We bound the expected assignment cost of each group $P_{j}$ with respect to $C'$, the center set returned by \RS. 
Let the centers selected in Phase $1$ of the algorithm are: $C_{f} \coloneqq \{ f_{1}^{*},\dotsc,f_{t}^{*} \}$, for $t = k \ln (2c \cdot n/\veps)$. In other words, $f_{i}^{*}$ is the center sampled in the $i^{th}$ iteration of Line $(3)$ of the subroutine. For each point $p \in P_j$ and iteration $i \in \{1,\dotsc,t\}$, we define a random variable $A^{i}_{p}$. It takes value $1$ if $p$ is unassigned after $i$ iterations; otherwise it is $0$. In other words, $A^{i}_{p} = 1$ if $p \in P_{u}$ after $i$ iterations. Given $A^{i}_{p} = 1$, the probability that $p$ is assigned to some center in $(i+1)^{th}$ iteration is: 

\begin{align*}
\textbf{Pr}[A^{i+1}_{p} = 0 \mid A^{i}_{p} = 1] &= \sum_{f \in F} \textbf{Pr}[\textrm{$p$ is assigned to $f$} \mid f^{*}_{i+1} = f] \cdot \textbf{Pr}[f^{*}_{i+1} = f ]
\end{align*}
\begin{equation}\label{equation:probability}
\hspace{1cm} = \sum_{f \in F} \frac{x_{f,p}}{y_{f}} \cdot \frac{y_{f}}{k} = \frac{\sum_{f \in F} x_{f,p}}{k} = \frac{1}{k} 
\end{equation}

\noindent The last equality follows from the second constraint of relaxed LP. Also, note that  $\textbf{Pr} [A^{1}_{p} = 0] = \frac{1}{k}$ using the above same analysis since the point was unassigned before the first iteration. 

Now, we show that the probability that $p$ is unassigned after $i$ iterations is $\left( 1 - \frac{1}{k}\right)^{i}$, i.e., $\textbf{Pr} [A^{i}_{p} = 1] = \left( 1 - \frac{1}{k}\right)^{i}$, for every $i \in \{1,\dotsc,t\}$. We prove this statement using induction on $i$:\\

\noindent \textbf{Base Case: } For $i = 1$, the probability that $p$ is unassigned after the first iteration is: \[
\textbf{Pr}[A_{p}^{1} = 1] = 1 - \textbf{Pr}[A_{p}^{1} = 0] = 1 -  \frac{1}{k}
\]

\noindent \textbf{Induction Step:} For $i > 1$, the probability that $p$ is unassigned after $i$ iterations is:
\begin{align*}
    \textbf{Pr} [A^{i}_{p} = 1]
    &= \textbf{Pr}[A^{i}_{p} = 1 \mid A^{i-1}_{p} = 1] \cdot \textbf{Pr}[A^{i-1}_{p} = 1] + \textbf{Pr}[A^{i}_{p} = 1 \mid A^{i-1}_{p} = 0] \cdot \textbf{Pr}[A^{i-1}_{p} = 0], \\
    &\quad \textrm{(using the law of total probability)}
\end{align*}

\noindent Note that $\textbf{Pr}[A^{i}_{p} = 1 \mid A^{i-1}_{p} = 0] = 0$ since $p \notin P_{u}$ after $(i-1)^{th}$ iteration; therefore it does not participate in $i^{th}$ iteration. Hence we get,
\begin{align*}
    \textbf{Pr} [A^{i}_{p} = 1] &= \textbf{Pr}[A^{i}_{p} = 1 \mid A^{i-1}_{p} = 1] \cdot \textbf{Pr}[A^{i-1}_{p} = 1] &&\\
    &= \left( 1 - \textbf{Pr}[A^{i}_{p} = 0 \mid A^{i-1}_{p} = 1] \right) \cdot \textbf{Pr}[A^{i-1}_{p} = 1] &&\\
    &= \left( 1 - \frac{1}{k} \right) \cdot \textbf{Pr}[A^{i-1}_{p} = 1] , &&\quad \textrm{(using Equation~(\ref{equation:probability}))}\\
    &= \left( 1 - \frac{1}{k} \right) \cdot \left( 1 - \frac{1}{k} \right)^{i-1}, &&\quad \textrm{(using Induction Hypothesis)}  \\
    &= \left( 1 - \frac{1}{k} \right)^{i} \\
\end{align*}
\noindent This proves that $\textbf{Pr} [A^{i}_{p} = 1] = \left( 1 - \frac{1}{k}\right)^{i}$, for every $i \in \{1,\dotsc,t\}$. \\

Now, we evaluate the expected assignment cost for each group $P_{j} \in \{P_{1},\dotsc,P_{\ell}\}$. Let $\alpha_{p}$ denote the assignment cost of each point $p \in P_j$ in the fractional optimal solution. That is, $\alpha_{p} = \sum_{f \in F} x_{f,p} \cdot d(f,p)^{\zl} \cdot w_{j}(p)$ . For every point $p \in P_j$ and center $f \in F$, let $E_{f,p}$ denote the event that $p$ is assigned to $f$ during Phase $1$. Let $\ep{E}{i}_{f,p}$ denote the event that $p$ is assigned to $f$ in the $i^{th}$ iteration, during Phase $1$. If $p$ remains unassigned after Phase $1$, then let $\beta_{p}$ denote the cost of $p$ during Phase $2$. Then, the expected cost of group $P_{j}$ with respect to the center set $C'$ is:

\begin{align*}
    \mathbb{E}[\costt{C',P_j}] &= \sum_{p \in P_j} \mathbb{E}[\costt{C',p}] \quad (\textrm{using linearity of expectation})\\
    &= \sum_{p \in P_j} \left( \, \sum_{f \in F} \textbf{Pr}[E_{f,p}] \cdot d(f,p)^{\zl} \cdot w_{j}(p) + \textbf{Pr}[A^{t}_{p} = 1] \cdot \beta_{p} \right)\\
\end{align*}
\vspace{-1cm}
\begin{equation}\label{equation:2}
    \hspace{3.1cm} = \sum_{p \in P_j} \left( \, \sum_{f \in F} \sum_{i = 1}^{t} \textbf{Pr}[\ep{E}{i}_{f,p}] \cdot d(f,p)^{\zl} \cdot w_{j}(p) + \textbf{Pr}[A^{t}_{p} = 1] \cdot \beta_{p} \right)\\
\end{equation} \\

\noindent Next, we show that $\sum_{f \in F} \sum_{i = 1}^{t} \textbf{Pr}[\ep{E}{i}_{f,p}] \cdot d(f,p)^{\zl} \cdot w_{j}(p)  \leq \alpha_{p}$ for every point $p \in P_{j}$.

\begin{align*}
    \sum_{f \in F} \sum_{i = 1}^{t} \textbf{Pr}[\ep{E}{i}_{f,p}] \cdot d(f,p)^{\zl} \cdot w_{j}(p) &\leq \sum_{f \in F} \sum_{i = 1}^{t} \textbf{Pr}[\, \ep{E}{i}_{f,p} \mid A^{i-1}_{p} = 1] \cdot \textbf{Pr}[A^{i-1}_{p} = 1]\cdot d(f,p)^{\zl}  \cdot w_{j}(p)\\
    &= \sum_{f \in F} \sum_{i = 1}^{t} \textbf{Pr}[\, \ep{E}{i}_{f,p} \mid A^{i-1}_{p} = 1] \cdot \left( 1-\frac{1}{k} \right)^{i-1} \cdot d(f,p)^{\zl} \cdot w_{j}(p) \\
    &= \sum_{f \in F} \sum_{i = 1}^{t} \frac{x_{f,p}}{y_{f}} \cdot \frac{y_{f}}{k} \cdot \left( 1-\frac{1}{k} \right)^{i-1} \cdot d(f,p)^{\zl} \cdot w_{j}(p) \\
    &= \sum_{i = 1}^{t}  \frac{\sum_{f \in F} x_{f,p} \cdot d(f,p)^{\zl}\cdot w_{j}(p) }{k} \cdot \left( 1-\frac{1}{k} \right)^{i-1} \\
    &= \sum_{i = 1}^{t}  \frac{\alpha_{p}}{k} \cdot \left( 1-\frac{1}{k} \right)^{i-1} \\
    &= \frac{\alpha_{p}}{k} \cdot \frac{ 1 - \left( 1-\frac{1}{k} \right)^{t} }{1/k}  \leq \alpha_p\\
\end{align*}

\noindent This proves that $\sum_{f \in F} \sum_{i = 1}^{t} \textbf{Pr}[\ep{E}{i}_{f,p}] \cdot d(f,p)^{\zl} \cdot w_{j}(p)  \leq \alpha_{p}$. Now, substituting this inequality in Equation~(\ref{equation:2}), we get

\begin{align*}
    \mathbb{E}[\costt{C',P_j}] &\leq \sum_{p \in P_j} \alpha_p + \sum_{p \in P_j} \textbf{Pr}[A^{t}_{p} = 1] \cdot \beta_{p} &&\\
    &= \sum_{p \in P_j} \alpha_p + \sum_{p \in P_j} \left( 1 - \frac{1}{k} \right)^{t} \cdot \beta_{p} &&\\
    &= \sum_{p \in P_j} \alpha_p + \left( 1 - \frac{1}{k} \right)^{t} \cdot O(\ell) \cdot OPT && \quad \textrm{($\because$ $O(\ell)$-approximation algorithm used in Phase $2$)}\\
    &\leq \sum_{p \in P_j} \alpha_p + O(\ell) \cdot OPT \cdot  \frac{\veps}{2c \cdot n} && \quad \textrm{($\because t = k \ln (2c \cdot n/\veps)$)} \\
    &\leq \sum_{p \in P_j} \alpha_p + c' \cdot \ell \cdot OPT \cdot  \frac{\veps}{2c \cdot n} &&\quad (\textrm{for some constant $c'>0$}) \\
    &\leq \sum_{p \in P_j} \alpha_p + c' \cdot OPT \cdot  \frac{\veps}{2c } &&\quad (\textrm{$\because \ell \leq n$}) \\
    &= \sum_{p \in P_j} \alpha_p + \frac{\veps}{2} \cdot OPT &&\quad (\textrm{we choose $c$ such that $c = c'$}) \\
    &\leq \gamma + \frac{\veps}{2} \cdot OPT &&\quad (\textrm{using constraint (3) of the relaxed LP}) \\
    &\leq \left( 1+\frac{\veps}{2} \right) \cdot OPT &&\\
\end{align*}
This completes the proof of Lemma~\ref{lemma:bicriteria}.

\section{Probability Amplification (Proof of Lemma~\ref{lemma:amplification})}\label{sec:amplification}

Here, we show that repeating the algorithm: \RS, sufficient times gives a $(1+\veps)$-approximation to the socially fair clustering problem with high probability. Suppose \RS \, is repeated $r = \frac{8 \ln n}{\veps}$ times, independently. Let $C'_{1},\dotsc,C'_{r}$ be the obtained center sets for each call to the algorithm. Then, we show that the center set $C \coloneqq C_{1}' \cup \dotsc \cup C_{r}'$ is a $(1+\veps)$ approximation to the optimal fair clustering cost of $P$, i.e., $\fair{C,P} \leq (1+\veps) \cdot OPT$, with probability at least $1-1/n$. The proof is as follows:

We say that a group $P_j$ \emph{violates} the fairness bound with respect to some center set $C'_{i}$, if $\costt{C'_{i},P_j} > (1+\veps) \cdot OPT$.
Now, note that for every center set $C'_{i} \in \{C'_1,\dotsc,C'_{r}\}$ and group $P_{j} \in \{ P_1\dotsc,P_{\ell} \}$, we have that $\mathbb{E}[\costt{C'_{i},P_j}] \leq (1+\veps/2) \cdot OPT$, using Lemma~\ref{lemma:bicriteria}. Furthermore, using the Markov's inequality, we get the following probability bound:
\begin{align*}
\textbf{Pr} \left[ \, \costt{C_{i}',P_j} > (1+\veps) \cdot OPT \,\right] < \frac{1+\veps/2}{1+\veps} = 1 - \frac{\veps/2}{1+\veps} \leq 1 - \frac{\veps}{4}, \quad \textrm{for $\veps \leq 1$}\\
\end{align*}

\noindent In other words, $P_j$ violates the fairness bound with respect to $C_{i}^{'}$ with probability at most $1 - \veps/4$. Then, the probability that $P_j$ violates the fairness bound with respect to every center set $C_{i}' \in \{C_1',\dotsc,C_{r}'\}$ is:

\begin{align*}
\textbf{Pr} \left[ \, \forall i \in \{1,\dotsc,r\}, \, \costt{C_{i}',P_j} > (1+\veps) \cdot OPT \,\right] &< \left( 1 - \frac{\veps}{4} \right)^{r}, \quad && \textrm{($\because$ Independent events)}\\
&\leq \frac{1}{n^2}, \quad && \Big( \textrm{ $ \because r = \frac{8 \ln n}{\veps} $} \Big)
\end{align*}
Since $C \coloneqq C_{1}' \cup \dotsc \cup C_{r}'$, we get
\begin{align*}
\textbf{Pr} \left[ \, \costt{C,P_j} > (1+\veps) \cdot OPT \,\right] \leq \textbf{Pr} \left[ \, \forall i \in \{1,\dotsc,r\}, \, \costt{C_{i}',P_j} > (1+\veps) \cdot OPT \,\right] < \frac{1}{n^2}
\end{align*}

\noindent In other words, $P_j$ violates the fairness bound with respect to $C$ with probability at most $1/n^2$. Then, the probability that at least one of the groups in $\{P_{1},\dotsc,P_{\ell}\}$ violates the fairness bound with respect to $C$ is: 

\begin{align*}
\textbf{Pr} \left[ \, \exists P_{j} \in \{P_1,\dotsc,P_{\ell}\}, \, \costt{C,P_j} > (1+\veps) \cdot OPT \,\right] &<  \frac{\ell}{n^2}, \quad \textrm{(using union bound)} \\
& \leq 1/n \quad (\because \ell \leq n)
\end{align*}

\noindent Therefore, the probability that none of the groups in $\{P_{1},\dotsc,P_{\ell}\}$ violate the fairness bound with respect to $C$ is:

\begin{align*}
\textbf{Pr} \left[ \, \forall P_{j} \in \{P_1,\dotsc,P_{\ell}\}, \, \costt{C,P_j} \leq (1+\veps) \cdot OPT \,\right] &\geq  1 - \frac{1}{n} 
\end{align*}

\noindent This gives the following probability bound on the fair clustering cost of $P$:
\begin{align*}
\textbf{Pr} \left[ \fair{C,P} \leq (1+\veps) \cdot OPT \,\right] &=
\textbf{Pr} \left[ \, \max_{j \in [\ell]} \Big\{ \costt{C,P_j} \Big\} \leq (1+\veps) \cdot OPT  \,\right] \\
&= \textbf{Pr} \Big[ \, \forall P_{j} \in \{P_1,\dotsc,P_{\ell}\}, \, \costt{C,P_j} \leq (1+\veps) \cdot OPT \,\Big] \\
&\geq 1 - \frac{1}{n}
\end{align*}

\noindent This proves that $C$ is a $(1+\veps)$ approximation to the optimal fair clustering cost of $P$ with probability at least $1-1/n$. This completes the proof of Lemma~\ref{lemma:amplification}.
\section{$O(\ell)$ Approximation Algorithm}\label{appendix:ell_approx}

Let $\mathcal{I} = (\X,P,P_1,\dotsc,P_{\ell},w_1,\dotsc,w_{\ell},F,d,k,z)$ be any instance of the socially fair clustering problem. For a center set $C \subseteq F$, the \emph{unconstrained} clustering cost of $P$ is simply the sum of the clustering costs of every point, i.e., $
\costt{C,P} \coloneqq \sum_{j = 1}^{\ell} \sum_{p \in P_j} d(C,p)^{\zl} \cdot w_{j}(p)
$, where $d(C,p)$ denotes the distance of point $p$ to the closest center in $C$. Let $OPT_{u}$ denote the optimal unconstrained clustering cost of $P$. That is,
$$
OPT_{u} \equiv \min_{C\subseteq F \textrm{ and } |C| = k}  \Big\{ \costt{C,P}  \Big\}
$$

There are various polynomial time $O(1)$-approximation algorithms for the unconstrained clustering problem~\cite{kmedian:1999_charikar,kmedian:2002_naveen,kmedian:2013_Svensson}. We show that these algorithms give an $O(\ell)$ approximation for the socially fair clustering problem. Formally, we state this result as follows:

\begin{lemma}
Let $\mathcal{I} = (\X,P,P_1,\dotsc,P_{\ell},w_1,\dotsc,w_{\ell},F,d,k,z)$ be any instance of the socially fair clustering problem. Let $C$ be an $O(1)$-approximate solution to the unconstrained clustering cost of $P$, i.e., $\costt{C,P} = O(1) \cdot OPT_{u}$. Then, $C$ is also an $O(\ell)$-approximation to the fair clustering cost of $P$, i.e., $\fair{C,P} = O(\ell) \cdot OPT$.
\end{lemma}

\begin{proof}
Let $C^{*}$ be an optimal center set for the socially fair clustering objective, and $C^{*}_{u}$ be an optimal center set for the unconstrained clustering objective. Then, we show that $OPT_{u} \leq \ell \cdot OPT$ using the following sequence of inequalities:

\begin{equation}\label{eq:3}
OPT_{u} = \costt{C_{u}^{*},P} \leq \costt{C^{*},P} = \sum_{j = 1}^{\ell} \costt{C^{*},P_j} \leq \ell \cdot \max_{j \in [\ell]} \Big\{ \costt{C^{*},P_{j}} \Big\} = \ell \cdot OPT 
\end{equation}

\noindent Now, we show that $\fair{C,P} = O(\ell) \cdot OPT$ using the following sequence of inequalities:

$$\fair{C,P} = \max_{j \in [\ell]} \Big\{ \costt{C,P_j} \Big\} \leq \sum_{j = 1}^{\ell} \costt{C,P_j} = \costt{C,P} = O(1) \cdot OPT_{u} = O(\ell) \cdot OPT$$.

\noindent The last equality follows from Equation~(\ref{eq:3}). This proves the lemma. \qed
\end{proof}
Note that the above result is simply a generalization of Theorem 1 of~\cite{fairness:2020_Socially_Ghadiri}.

\section{Reduction: Set Coverage to k-Supplier}\label{appendix:reduction}

The set coverage problem is defined as follows.
\begin{definition}[Set Coverage]
Given an integer $k>0$, a set $U$, and a collection $\mathscr{C} = \{S_1,\dotsc,S_m\}$ of subsets of $U$, i.e., $S_j \subseteq U$ for every $j \in [m]$, determine if there exist $k$ sets in $\mathscr{C}$ that cover all elements in $U$.
\end{definition}

Now we describe the reduction. Given a set coverage instance $(U,\mathscr{C},k)$, we construct a $k$-supplier instance $(P,F,d,k)$ as follows. For every set $S_{i} \in \mathscr{C}$, we define a center $c_{i} \in F$. For every element $e \in U$, we define a point $x_{e} \in P$. Let us define the distance function $d(.,.)$ as follows. For any two points $x_{e},x_{e'} \in P$, or $c_{i},c_{j} \in F$, the distance $d(x_{e},x_{e'}) = d(c_{i},c_{j}) = 2$. For any point $x_{e} \in P$ and $c_{i} \in F$, if $e \notin S_{i}$, the distance $d(x_e,c_{i}) = 3$; otherwise $d(x_e,c_{i}) = 1$. Furthermore, assume that $d(.,.)$ is a symmetric function, i.e., $d(x,y) = d(y,x)$ for every $x,y \in P \cup F$. Also assume that $d(x,x) \geq 0$ for every $x \in P \cup F$. It is easy to see that $d(.,.)$ satisfies all the properties of a metric space.

Now, suppose that there exist $k$ sets: $S_{i_{1}},\dotsc,S_{i_{k}}$ in $\mathscr{C}$ that cover all elements of $U$, i.e., $S_{i_{1}} \cup \dotsc \cup S_{i_{k}} = U$, then the center set $C = \{c_{i_{1}},\dotsc,c_{i_{k}}\}$ gives the $k$-supplier cost $1$. On the other hand, if there does not exist any $k$ sets in $\mathscr{C}$ that could cover all elements of $U$, then for any center set $C \subseteq F$ of size $k$ there would exist a point $x \in P$ at a distance of $3$ from $C$, i.e., $d(C,x) = 3$. Therefore, the $k$-supplier cost would be $\rp{3}{\z}$. Since the set coverage problem is $\mathsf{NP}$-hard, it implies that the $k$-supplier problem can not be approximated to any factor better than $\rp{3}{\z}$, in polynomial time, assuming $\mathsf{P} \neq \mathsf{NP}$. Moreover, the following FPT hardness result is known for the set coverage problem (see Corollary 2 and Theorem 25 of~\cite{fpt:2020_Pasin_SODA}).

\begin{theorem}
For any function $g \colon \mathbb{R}^{+} \to \mathbb{R}^{+}$, there is no $g(k) \cdot n^{o(k)}$ time algorithm for the set coverage problem, assuming Gap-ETH.
\end{theorem}

\noindent This implies the following FPT hardness of approximation for the $k$-supplier problem.

\begin{theorem}\label{theorem:hardness_supplier2}
For any $\veps>0$ and any function $g \colon \mathbb{R}^{+} \to \mathbb{R}^{+}$, the $k$-supplier problem can not be approximated to factor $(\rp{3}{\z} - \veps)$ in time $g(k) \cdot n^{o(k)}$, assuming Gap-ETH.
\end{theorem}

\end{document}